\documentclass{article}
\usepackage{amssymb,bm}
\usepackage{amsmath}
\usepackage{amsthm}
\usepackage{amsfonts}
\usepackage{mathrsfs}
\usepackage{appendix}
\usepackage{hyperref}
\usepackage{authblk}
\usepackage{color}
\usepackage{optidef}
\usepackage{graphicx}
\usepackage{algpseudocode}
\usepackage{ifthen}
\newtheorem{lemma}{Lemma}

\theoremstyle{remark}

\usepackage{graphics}
\usepackage{graphicx}
\usepackage{geometry}
\usepackage{algorithm}
\hypersetup{colorlinks,linkcolor={blue},citecolor={blue},urlcolor={red}}

\usepackage[english]{babel}

\usepackage{cite}
\usepackage{threeparttable}
\usepackage{float}
\floatstyle{plaintop}
\restylefloat{table}

\floatstyle{plain}
\restylefloat{figure}
\usepackage{placeins}

\begin{document}

\title{Optimal Adversarial Testing: Extracting Honest Test Results from Dishonest Test Takers}


\author{Owen Cox{\footnote{Department of Electrical and Computer Engineering, University of Iowa, Iowa City, IA 52242.}}\,~~
April Xu{\footnote{ Iowa City Math Club.
}}\,~~
 Weiyu Xu {\footnote{
Department of Electrical and Computer Engineering, University of Iowa, Iowa City, IA 52242. } \,~~
\footnote {Owen Cox helped coded and ran the numerical experiments, April Xu derived (\ref{mainformula}),  helped with separate codes for verifications, ran the  numerical experiments and prepared references and formulas in LaTex. Weiyu Xu supervises the research. }
}}
\maketitle

\begin{abstract}
In applications, it is often required to test objects or people to determine their qualities in terms of certain metrics. However, besides being naturally noisy, the test results can be corrupted by adversarial behaviors of objects or people being tested (test takers). For example, dishonest test takers can cheat in the exams to distort the test results. With the development of AI technologies, such distortions driven by cheating using AI technologies are becoming more commonplace and severe. In this paper, we propose optimal testing strategies which can still recover needed test results even if there are cheaters polluting the results.  The proposed testing strategies will optimally re-test selected group of test takers using different testing security measures. We determine the optimal testing strategies using a dynamic programming method.

\end{abstract}



\section{Introduction}
In many applications, it is often required to test objects or people to determine their qualities in terms of certain metrics. That would require sampling or testing each object or person to obtain test or sampling results of them. 

However, the test results can be noisy since they may be affected by natural perturbations or noises. In addition, besides being naturally noisy, the test results can be corrupted by adversarial behaviors of objects or people being tested (test takers). For example, dishonest test takers can cheat in the exams to distort the test results. With the development of AI technologies, such distortions driven by AI-aided cheating are becoming more commonplace and more severe \cite{DJEBBARI_2025, AIbasedcheating, fer333, article_1553831}. In order to make the final testing results faithful, the testing or sampling of the objects or people must be done in a way such that the test organizers are able to thwart the adversarial actions of tested objects or people. 

There may be different goals of testing, for example, determining the admission of students to a program or a collage, or determining the ranking of persons in a competition. In this paper, we consider testing which aims to select a number of top students. Suppose that there are $n$ people that we need to test. At the end of the tests, we need to select the top $k$  ($k<n$) people, say, to declare them as winners out of the $n$ people. For example, in the American Mathematics Competition (AMC) series organized by Mathematical Association of America (MAA), the organizers aim to select (approximately) $60$ top performing students to participate in the Mathematical Olympiad Program (MOP) each year, out of around 300,000 students \cite{bajnok2024amcmatters}. Or even one step further, the AMC organizers aim to select the USA national teams out of the approximately $300,000$ students to attend international math competitions including EGMO (European Girls' Mathematical Olympiad), RMM (Romanian Master of Mathematics), and IMO (International Mathematical Olympiad). We note that such scenarios happen in many disciplines, such as selection of physics, chemistry, biology, and informatics teams, or even in college exams. 

The successful selection of the top $k$ students depends on the faithfulness of the test results, which in turn depends on the honesty or integrity of test takers. In an alarming trend, some dishonest test takers use AI tools or other cheating mechanisms to cheat in exams, leading to unfaithful test results.  Contrary to traditional ways of cheating, the new ways of cheating using AI tools are very effective in increasing cheaters' test scores. Moreover, such cheating aided by AI technologies is more stealthy, and even harder to catch and detect  \cite{DJEBBARI_2025, AIbasedcheating, fer333, article_1553831}. Cheating test takers can go to various lengths to cheat: some may simply copy from neighbors; some may steal the exam contents beforehand by purchasing them or by computer system hacking; some may purchase answers from corrupt human coaches; some may use concealed or planted cell phones to access AI technologies to answer these questions; and some even use more stealthy AI glasses to cheat.

To defeat such cheating, it would require test organizers having much stricter proctoring, such as using metal detectors to search for any concealed electronic devices including cell phones and smart glasses on each test taker; or to make the test sites electromagnetic (EM) signal proof; or to carefully check the test sites to prevent cheaters from planting AI devices including cell phones in hidden places (note  that cheaters can avoid proctors and use electronic devices to access AI tools in bathroom break).  Unfortunately, it may be unrealistic or too expensive to implement these measures on a large scale. Sometimes local laws and policies do not even allow some of the strict-proctoring measures at exam sites.  When the number of students can be large, it will be expensive to implement strict proctoring on a large scale. For example, it may be simply impossible to adopt the strictest proctoring measures for all the 300,000 students taking part in the AMC tests.  

This raises the following question: how to achieve the goal of selecting the top $k$ students \emph{accurately}, \emph{economically} and \emph{efficiently} even if there may be many cheaters in the tests?

In this paper, we propose optimal testing strategies which can still faithfully achieve the testing goals even if there are cheaters polluting the raw testing results.  The proposed testing strategies will optimally re-test selected group of test takers using different levels of testing security measures. We determine the optimal (re)testing strategies using a dynamic programming method. We also consider a game-theoretic setting where the adversarial test takers can allocate cheating resources according to the testing strategies of test organizers. Our proposed approach can be extended to a continuous-valued version (when the number of tested people or objects goes to infinity) of the considered problem.  To the best of our knowledge, this is the first paper which formulates the adversarial testing problem in terms of AI cheating and computes the optimal re-testing strategy when there are adversarial test takers using AI technologies.

\section{Mathematical Model}\label{Sec:mathmodelMarkov}

Suppose that there are $n$ people that we need to test. The organizers have a metric $M(T, A)$ for evaluating the effectiveness of the testing strategy, where  $M(T, A)$ is a function of testing strategy $T$ and the test takers' adversarial cheating strategy $A$. We let the set of testing strategies be $\mathcal{T}$, and the set of test takers' cheating strategies be $\mathcal{A}$. Note that the testing strategy needs to specify:  how many tests are given; which set of people are respectively in each test; and how secure each test should be made to be. Note that the test organizers can adaptively determine the testing actions based on the results of the tests that have already happened.

The test organizer tries to maximize $M(T, A)$  while the adversarial test taker tries to minimize $M(T, A)$. Thus the optimal testing problem becomes a max-min problem: 
$$ \max_{T\in \mathcal{T}} \min_{A\in \mathcal{A}} M(T,A) .$$

In this paper, in particular, we consider the metric of successfully selecting the true $k$ top students out of $n$ students, where $k<n$. If the selection is successful, the metric is $1$; otherwise, the metric is $0$.  We assume that, during each test, the same cost is used in proctoring each student. Moreover, for each test,  a student has a student-specific ``cheating'' budget amount all of which the student can use for cheating in each single test. We assume that if the proctoring cost spent on each student exceeds a function of a student's budget/capability for cheating, an accurate testing result will be obtained for the student, and, if the student is not among the top $k$ students, the student will be eliminated or disqualified for the selection purpose. The optimal testing problem considered above thus translates into a problem of designing the optimal strategy which can select the top $k$ students correctly and also incurs the smallest amount of total testing cost.  Please note that the test organizers do not know the exact proctoring cost for each individual student that is needed to produce a faithful result, but the organizers know the distribution of such costs over the set of all the students.

Instead of performing the most strict tests for all the students directly, we propose to selectively administer tests at different security levels. The students who pass earlier tests with potentially lower security levels (tests with lower testing cost per student) will advance to be re-tested with higher security levels (tests with higher testing cost per student). Using the results obtained at different security levels administered for different groups of students, we can select the top $k$ students accurately and efficiently. To find the optimal testing strategy, we will need to determine how many tests we will perform and what will be the security level for each test.  

To describe the algorithm for finding the optimal test strategy, we first specify some notations and basic assumptions. 
\begin{itemize}
\item $n$ is the number of students. $k$ is the number of top students we want to select out of these $n$ students.
\item Let $a_t$ be the per-person proctoring cost in the $t$-th test.
\item The proctoring cost per person is related to the exam security level: the higher such cost is, the higher is the exam security level.
\item Let the cost per person for proctoring the test be $x$. Let $f(x)$ be the number of students that pass a test when the per-person cost for proctoring is $x$.
\item Let $g(i)$ be the needed per-person proctoring cost which is just enough to eliminate student with index $i$. 
\item We rank the students in an increasing order of $g(i)$. Namely, the bigger $i$ is, the larger is $g(i)$.
\item A person who does not rank among the top $k$ students passes a test if and only if the exam security level is below a certain person-specific threshold $g(i)$: if a person is honest, that threshold (namely $g(i)$) is low; while if a person is dishonest, that threshold (namely $g(i)$) is high. 
\item A person who ranks among the top $k$ students will always pass any exam no matter what the exam security level is (since the exam is designed for the top $k$ students to pass).    
\item $a_0=0$, namely before any test starts, no cost is incurred. 
\item $f(a_0)=n$, namely initially there is no test, so the number of remaining students is $n$. 
\end{itemize}

We remark that our mathematical model can be extended beyond testing using different security levels. For example, $g(i)$ may simply model testing cost needed to eliminate the $i$-th person for consideration in any application scenario. That may include designing tests with different levels of problem difficulty.  We further remark that the test organizers do not know the exact $g(i)$ number for any specific student: they only know the $n$ values of $g(i)$'s or the distribution of $g(i)$'s, but do not know the correspondence between these $n$ values and each student. 

In fact, mathematically, this problem is about sampling theory \cite{kutyniok2012theoryapplicationscompressedsensing}. In this problem, we have $n$ objects to test.  Each object, say, object $i$, has a tolerance threshold $g(i)$. When the test organizers test object $i$ with a testing cost $c$, the testing result will be $1$ (namely pass) if $c\geq g(i)$ but will be $0$ (namely fail) if $c< g(i)$. The mathematical question just becomes: how to design the tests such that we can use the smallest possible total testing costs to find the $k$ objects with the largest $g(i)$'s? Note that the sampling is adaptive: one can determine the design of future tests using the testing results obtained in previous tests. Moreover, tests must be administered in batches. In each batch, all the passing students from earlier tests must be tested using the same test cost $c$.

\section{Optimal testing strategy}

We now analyze the optimal testing strategy's properties and give an efficient algorithm which can find the optimal strategy. The algorithm is of complexity $O(n^2)$, and the complexity can possibly be further reduced. 

\subsection{Total testing costs}

Suppose that we do only one test with per-person cost as $a_1$, the total testing cost will be 
$$f(a_0)a_1.$$
Similarly, if we do only two tests with two different levels of security measures, the total cost will be
$$f(a_0)a_1+f(a_1)a_2.$$
If we do only three such tests, the total cost will be
$$f(a_0)a_1+f(a_1)a_2+f(a_2)a_3.$$
Generalizing this, if we do $T$ tests with $T$ different levels of security measures, the total cost will be



\begin{equation}
\label{mainformula} \sum_{t=1}^{T} f(a_{t-1})a_t.    
\end{equation}

We will try to determine $T$, and all the per-person cost $a_0$, $a_1$, ..., and $a_{T}$ to minimize the total cost. 

\subsection{A key lemma for dynamic programming}

We have the following key lemma specifying the structure of the optimal testing strategy. 

\begin{lemma}
Suppose that the testing measures $a_1$, $a_2$, ..., $a_{m}$, ..., and $a_{T}$ are a set of sequential optimal test measures to select the top $k$ students out of $n$ students. Then the $a_{m+1}$, ... $a_{T}$ must be the optimal testing strategy for selecting the top $k$ students out of the  $f(a_{m})$ students who pass the test under security measure $a_{m}$.
\end{lemma}

\begin{proof}
Based on the general formula, 
\[ \sum_{t=1}^{T} f(a_{t-1})a_t,\]
when $a_1$, $a_2$, ..., and $a_m$ are given, the 
$a_{m+1}$, ... $a_{t}$
must optimize the following expression
\[ \sum_{t=m+1}^{t'} f(\tilde{a}_{t-1})\tilde{a}_t,\]
where $t'$ may not be equal to the original $T$ used in the original optimal testing strategy $a_1$, $a_2$, ..., and $a_T$; $\tilde{a}_{m}=a_{m}$ and $f(a_{t'})=k$.

We claim that \[ A=\sum_{t=m+1}^{t'} f(\tilde{a}^*_{t-1})\tilde{a}^*_t\] is equal to \[ B=\sum_{t=m+1}^{T} f(a_{t-1})a_t,\] where $\tilde{a}^*$'s mean the per-person cost in the corresponding optimal solution minimizing $A$. 

First of all, $A\leq B$ since $\tilde{a}^*$'s form the optimal solution. Secondly, if $A<B$, we can just replace the original sequence of measures $a_{m+1}$, ... $a_{T}$ with $\tilde{a}^*$'s, and this gives a smaller value for \[ \sum_{t=1}^{T} f(a_{t-1})a_t,\] contradicting the assumption that the testing measures $a_1$, $a_2$, ..., $a_{m}$, ..., and $a_{T}$ are a set of optimal test measures.  Thus,  $a_{m+1}$, ... $a_{t}$ also form the optimal testing strategy for selecting the top $k$ students out of the  $f(a_{m})$ students who pass the test under security measure $a_{m}$.   

\end{proof}

\subsection{Algorithm for finding the optimal testing strategies}

We first present the key idea and the higher-level pseudocode for our algorithm  in Algorithm \ref{alg:psuedo_ot}.  The key idea is a dynamic programming method, where we can use the optimal testing strategy of selecting $k$ students out of $l$ students in determining the optimal strategy of selecting $k$ students out of $m$ students, where $m>l$. More specifically, suppose that we are facing the task of selecting $k$ students out of $m$ students. Suppose the first test conducted for such a task is just secure enough to eliminate $q$ lowest-ranking students (ranked in terms of the required per-person test cost to eliminate the test-taker for further consideration), then the remaining tests form an optimal testing strategies for the task of selecting $k$ students out of the $(m-q)$ students, based on the key lemma mentioned above.  We give our higher-level pseudocode as follows. Note that the top $l$ test takers in the descriptions mean the test takers with the top $l$ per-person proctoring cost $g(i)$'s.

\begin{algorithm}
\caption{High-level explanation: optimal testing strategy finding algorithm}\label{alg:psuedo_ot}
\begin{algorithmic}[1]
\Require $n$, $k$, $g(i)$, vector ``$cost$'' of $n-k+1$ elements, vector ``$next-index$'' of $n-k$ elements. The indices go from $1$ to the length of the vectors.
\Ensure Optimal testing strategy with the smallest testing cost.
\State List the $n$ test takers in an increasing order of their respective values $g(i)$'s

\State Initialize the optimal cost for selecting the top $k$ test takers out of these $k$ test takers as $0$ (no need to do any testing). Namely, set $cost[n-k+1]\gets 0$
\State $i \gets n-k$
\While{$i \geq 1$} 
\Comment{ Comments: ``This is for computing the the optimal cost for selecting the top $k$ honest test takers out of the top $n-i+1$  test takers''} 

\For{$1 \leq j \leq (n-i+1-k)$}

Compute the total testing cost by letting the first test's cost be just enough so that the

lowest ranked $j$ test takers of the top $(n-i+1)$ test takers are eliminated. 
\EndFor

Pick the minimum cost among the computed $(n-i+1-k)$ costs above. Set that minimum

cost as the smallest possible cost for selecting the top $k$ test takers out of the top $(n-i+1)$ 

test takers. 
\State $i\gets i-1$
\EndWhile
\end{algorithmic}
\end{algorithm}

Then we give the detailed algorithm as follows in Algorithm \ref{alg:ot}. 
\begin{algorithm}
\caption{Optimal testing strategy finding algorithm}\label{alg:ot}
\begin{algorithmic}[1]
\Require $n$, $k$, $g(i)$, vector ``$cost$'' of $n-k+1$ elements, vector ``$next-index$'' of $n-k$ elements. The indices go from $1$ to the length of the vectors.
\Ensure Optimal testing strategy with the smallest testing cost.
\State $cost[n-k+1]\gets 0$
\State $i \gets n-k$
\While{$i \geq 1$}
\For{$1 \leq j \leq n-k-i+1$}

\If{$g[n-k-j+1]\times(n-i+1)+cost[n-k-j+2]< cost[i]$}
    \State next-index$[i] \gets n-k-j+2$ 
\EndIf
\State cost$[i] \gets \min \{cost[i], g[n-k-j+1]\times(n-i+1)+cost[n-k-j+2]\}$
\EndFor
\State $i\gets i-1$
\EndWhile
\end{algorithmic}
\end{algorithm}
Here $next-index[i]$ means that the optimal testing for selecting the top $k$ students out of students indexed from $i$ to $n$ will first use the test which is just enough to eliminate students with index smaller or equal to $next-index[i]-1$; and then we are left with the task of selecting the top $k$ students out of students indexed from $next-index[i]$ to $n$. Here $cost[i]$ denotes the cost needed to select the top $k$ students out of students with indices $i$ (inclusive) to $n$.

\section{Numerical results}
\label{sec:stationaryfor3agent}
 In this section, we present numerical results of the optimal testing strategy. 

In the first example, we consider the job of selecting $60$ top performers out of $300,000$ students.  For the top 60 performers,  no matter how secure the test is, they will always pass the exam (since the test is designed at a level such that the top 60 students pass). For $298000$ students among these $300,000$ students, the needed per-person proctoring cost will be 10 dollars, since these students are honest and will not even consider cheating under normal proctoring. For another $1000$ students, the needed per-person proctoring cost to counter and defeat their cheating strategies is $100$ dollars. For another $900$ students, the needed per-person proctoring cost to counter and defeat their cheating strategies is $200$ dollars. For another $40$ students, the needed per-person proctoring cost to counter and defeat their cheating strategies is $4000$ dollars. Algorithm \ref{alg:ot} gives the following optimal test strategies: there will be 3 rounds of tests in total, the first-round test incurs per-person proctoring cost of $10$ dollars for $300000$ students, the second-round test incurs per-person proctoring cost of $200$ dollars for $2000$ students, and the third-round test incurs per-person proctoring cost of $4000$ dollars for $100$ students. The optimal total testing cost is thus $3800000$ ($3.8$ million) dollars.  In contrast, if we just directly use the strictest proctoring measure in the fist place, we would need $4000\times 300000=1200000000$  ($1.2$ billion) dollars: in terms of total test cost, the optimal testing strategy is only around $\frac{3}{1000}$, or $0.3\%$ fraction of this naive approach, thus saving significant test cost without sacrificing the final accuracy of the selection.

In a second example, we consider a case of selecting 100 top performers from 50,000 students. For 20 students outside those 100 students, the needed proctoring cost will be 200 dollars per student. For 40 students, the needed proctoring cost will be 100 dollars per student; for 500 students, the proctoring cost will be  50 dollars per student; for 900 students, the proctoring cost will be 25 dollars per student; for 10000 students, the proctoring cost will be 20 dollars per student, and for the remaining number of students, the needed proctoring cost will be 10 dollars per person. The optimal cost for this strategy according to Algorithm \ref{alg:ot} will be 835200 dollars, while the naive cost strategy of directing testing with the highst per-person cost will be 10000000 dollars. Therefore, the optimal strategy is only 8.4\% the cost of the naive strategy.
The optimal strategy in this case involves five tests. After the first test, 11560 students remain. After the second, 1560 remain, after the third, 660 remain. After the fourth test, 160 students remain, and after the final test, we have selected the top 100 performers.

In the third example, we consider the job of selecting $60$ top performers out of $30,000$ students.  For the top 60 performers,  no matter how secure the test is, they will always pass the exam (since the test is designed at a level such that the top 60 students pass). For $27000$ students among these $300,000$ students, the needed per-person proctoring cost will be 10 dollars, since these students are honest and will not even consider cheating under normal proctoring. The formula for the needed cost of a proctored test to eliminate cheaters for student $i$, where $i$ ranges from $27001$ (inclusive) to $29940$ is $10\times (1+\frac{(i-27000)^2}{10^5})$. In this optimal testing strategy,  the total number of tests done is $4$: the first test will disqualify 27007 students by using per-person test cost of 10.004900000000001  (copied from the code to distinguish it from 10) dollars; the second test will further bring the total number of disqualified students to 27549 (namely disqualifying $542$ more students) by using  per-person test cost of 40.1401 dollars; the 3rd test will further bring the total number of disqualified students to 28784 (namely disqualifying 1235 more students) by using  per-person test cost of 328.2656 dollars; and the 4th test will further bring the number of disqualified students to $29940$ (namely disqualifying 1156 more students) using per-person test cost of $874.3600000000001$ dollars.  The total testing costs will be 228,8085 dollars, only $8.7722894459\%$  of the total cost of the naive approach of using strictest proctoring for every student. 

However, if we change the above problem's corresponding expression to 
$10\times(1+\frac{(i-27000)^2}{10^6})$, the optimal testing strategy consists of three tests: the first test will disqualify 27063 students using a per-person test cost of 10.03969 dollars; the 2nd test will disqualify 28645 students using a per-person test cost of 37.060249999999996 dollars; and the 3rd test will disqualify $29940$ students  using a per-person test cost of 96.43599999999999 dollars.  The total testing costs will be 540705 dollars, only $18.7\%$  of the total cost of the naive approach of using strictest proctoring for every student.





\section{Conclusions, discussions and future directions}
\label{sec: conclusions}

In applications, it is often required to test objects or people to determine their qualities in terms of certain metrics. However, besides being naturally noisy, the test results can be corrupted by adversarial behaviors of objects or people being tested (test takers). For example, dishonest test takers can cheat in the exams to distort the test results. With the development of AI technologies, such distortions are becoming more commonplace. In this paper, we propose optimal testing strategies which can still recover needed test results even if there are cheaters polluting the results.  The proposed testing strategies will optimally re-test selected group of test takers using different testing security measures. We determine the optimal testing strategies using a dynamic programming method.  In future works, we will extend our work with a continuous version of the testing problem where the number of tested students is large and the portion of students that needs a per-person proctoring cost of $x$ can be described by a (possibly-continuously-valued) probability density function of $p(x)$ which is defined over continuously-valued  $x$.  In this paper, we consider faithfully figuring out the top $k$ students; in future work, we can extend this by considering recovering other information faithfully, such as the ranking of top students or the ranking of the whole group of students. Note in this current paper, while we formulate the problem generally, we only investigate in detail the setting where a cheater has determined resource or budget for cheating in an exam. Thus our work can be extended to cases where the cheaters have more constraints on allocating their resources to different tests. For example, the cheaters have a constraint on the total budget that can be used for all the levels of tests. Thus it would also be interesting to extend this work to other sets of feasible actions for the cheaters.  We can also extend this paper to settings where the performances of test takers are random instead of deterministic. 

We remark that our proposal of re-testing using different levels of security measures is different from different levels of competitions of AMC series such as AMC, AIME, and USA(J)MO, which are more intended to test different aspects or levels of math strengths, such as skills in basic or advanced computational problems, or proof problems. Note that tests of different levels of security measures can be integrated with or separated from tests evaluating different levels or aspects of math strengths.  For example, if we want to select the USA(J)MO qualifiers from the AIME test (USA(J)MO qualification is also often considered a great honor for a student. There is also a need to do that selection accurately.), the test organizers can re-test students with AIME tests using different levels of security measures as optimally calculated by this paper. Note that normally the organizers aim to select around $500$ students out of around $5000$ AIME qualifiers to advance to USA(J)MO qualification. The preliminary AIME with a lower level of security measure can reduce the number of potential such candidates for USA(J)MO selection to around $1000$ students (when doing so, we conservatively assume that there are $500$ potential cheaters in the top $1000$ students emerging from the AIME of lower-level security, which is likely extreme). That $1000$ number is small enough such that all the $1000$ potential candidates can easily fit into test centers with way more strict security measures and can take another more-strictly-proctored AIME test across the USA within a uniform time slot. This will alleviate the concern of test center capacity and also the concern of students having to take the test in different time slots, compared with the strategy where we directly test all of the around $5000$ AIME qualifiers in high-security test centers. More importantly, this can optimally reduce the total proctoring costs while achieving faithful selection results.

\bibliography{Reference_20260213, Reference_addedagents,Reference_NullSpaceCondition}

\begin{thebibliography}{1}

\bibitem{fer333}
Olalekan~J. Akintande.
\newblock Artificial versus natural intelligence: Overcoming students' cheating likelihood with artificial intelligence tools during virtual assessment.
\newblock {\em Future in Educational Research}, 2(2):147--165, 2024.

\bibitem{bajnok2024amcmatters}
Bela Bajnok.
\newblock The {AMC} -- {W}hat it is and {W}hy it matters, 2024.

\bibitem{DJEBBARI_2025}
Zakia DJEBBARI.
\newblock A new era of exam cheating: The use of {ChatGPT} and the rise of academic dishonesty.
\newblock {\em ALTRALANG Journal}, 7(2):431–439, 2025.

\bibitem{kutyniok2012theoryapplicationscompressedsensing}
Y.C. Eldar and G.~Kutyniok.
\newblock {\em Compressed Sensing: Theory and Applications}.
\newblock Cambridge University Press, 2012.

\bibitem{AIbasedcheating}
Sandra Leaton~Gray, Dominic Edsall, and Dimitris Parapadakis.
\newblock Ai-based digital cheating at university, and the case for new ethical pedagogies.
\newblock {\em Journal of Academic Ethics}, 23:2069--2086, 05 2025.

\bibitem{article_1553831}
Zekeriya NartgÃn and Eugene Kennedy.
\newblock Cheating in higher education in the age of artificial intelligence.
\newblock {\em International Journal on Lifelong Education and Leadership}, 10(2):47--54, 2024.

\end{thebibliography}

\end{document}